\documentclass{article}
\usepackage{prem}

\title{Pivot based correlation clustering in the presence of good clusters}
\author{ 
{David Rasmussen Lolck\thanks{University of Copenhagen, Email: \texttt{dalo@di.ku.dk}. Supported by VILLUM Foundation Grant 54451, Basic Algorithms Research Copenhagen (BARC)}} 
\and
{Mikkel Thorup\thanks{University of Copenhagen, Email: \texttt{mthorup@di.ku.dk}. Supported by VILLUM Foundation Grant 54451, Basic Algorithms Research Copenhagen (BARC)}}
\and
{Shuyi Yan\thanks{University of Copenhagen, Email: \texttt{shya@di.ku.dk}. Supported by VILLUM Foundation Grant 54451, Basic Algorithms Research Copenhagen (BARC).}}
}
\date{\today}

\begin{document}

\maketitle
\begin{abstract}
    The classic pivot based clustering algorithm of Ailon, Charikar and Chawla [JACM'08] is factor 3, but all concrete examples showing that it is no better than 3 are based on some very good clusters, e.g., a complete graph minus a matching. By removing all good clusters before we make each pivot step, we show that this improves the approximation ratio to $2.9991$. To aid in this, we also show how our proposed algorithm performs on synthetic datasets, where the algorithm performs remarkably well, and shows improvements over both the algorithm for locating good clusters and the classic pivot algorithm.
\end{abstract}

\section{Introduction}
Correlation clustering is a fundamental clustering problem. The problem is modelled on an unweighted, undirected graph $G=(V,E)$, where the vertices correspond to the objects that are being clustered while the edges are pairwise similarity statements. In this setting, a clustering is then a partitioning of the graph, and how good a clustering is, is given by the number of pairs of vertices sharing an edge that was put in different clusters plus the number of pairs of vertices that was put in the same cluster without sharing an edge.\footnote{In some literature on correlation clustering it is common to model this problem as a complete graph, where each edge is labelled with either $+$ or $-$. The two versions are equivalent, though since we are dealing with near-linear running times in the number of edges, the first formulation is in our case more advantageous as $+$ and $-$ edges start to lose the symmetry in their meaning to the algorithm.}

\subsection{Prior work}

The problem of correlation clustering was first introduced by Bansal, Bloom and Chawla \cite{BBC04}. They gave a constant approximation to this problem. This was since improved by Ailon, Charikar and Newman \cite{ACN08} who presented a $3$-approximate combinatorial linear time algorithm, called the pivot algorithm and a $2.5$-approximate LP-based algorithm.

The pivot algorithm in particular has had a lot of influence on the development of algorithms for correlation clustering due to its simplicity. The algorithm works by randomly selecting a vertex, called the pivot, and then clustering it and all of its neighbours together. Then repeat on the remaining graph.

Following the work of \cite{ACN08}, Chawla, Makarychev, Schramm, and Yaroslavtsev gave a near-optimal rounding for the standard LP relaxation achieving a $2.06$ approximation \cite{CMSY15}. More recently, Cohen-Addad, Lee, and Newman broke the integrality barrier of $2$ for the standard LP by leveraging a constant number of rounds of the Sherali--Adams hierarchy, obtaining an $(1.994+\varepsilon)$-approximation \cite{CLN22}. Very recently, Cohen-Addad, Lee, Li, and Newman further improved the approximation guarantee to $1.73$ via new rounding ideas based on preclustering and handling correlated rounding error \cite{CLLN23+}. After this work, a new LP framework was introduced by Cao, Cohen-Addad, Lee, Li, Newman, and Vogl, called the \emph{cluster LP}, which lead to an approximation ratio of $1.485+\eps$ \cite{cao2024understanding}, which is the currently best known approximation ratio.

In addition to the work on the approximation ratio, there have also been a lot of work on the running time of the clustering algorithms. A lot of the previous works that improved the approximation ratio has been focused on LP-based approaches, where running times are only really considered in-so-far they are polynomial. That said, there have been significant progress in this direction. Initially, the original pivot algorithm \cite{ACN08} trivially had a running time of $O(m)$ with a $3$-approximation. A different algorithm was subsequently presented by Assadi and Wang that achieved a constant approximation ratio in sublinear time $O(n \log^2 n)$. Building on this algorithm as a preprocessing step, an $1.847+\eps$ approximate solution in time $O(n\log^2 n)$ was presented by Cohen-Addad, Lolck, Pilipczuk, Thorup, Yan and Zhang \cite{CLMTYZ24}. Finally, the authors of \cite{cao2025fastLP} showed that it is possible to match the best known approximation ratio of $1.485+\eps$ even in the sublinear time setting. 

The problem however with the more recent results \cite{CLMTYZ24,cao2025fastLP} is that while from a theoretical perspective, they achieve these results under some very high constant terms and exponential dependencies on $\eps$, making them unsuitable for practical applications. So from this perspective, from the results so far presented, the pivot algorithm is the most practical algorithm so far presented in this paper from the perspective of theoretical guarantees.

\subsection{Technical Overview}
In this work, we will be looking at presenting a new algorithm for correlation clustering, which both has a nice approximation ratio in theory while also performing well in practice. The way we will do this is by showing an algorithm which can be seen as a combination of the pivot algorithm by \cite{ACN08} which is a $3$ approximation, and the atom-based clustering algorithm by \cite{CLMNP21}, which is only a constant approximation. %

The goal is to use the clusters that the atom-based clustering algorithm would find, but further process them such that we also include some vertices that belong to these clusters.  Whenever we are unable to find any atoms, we will instead run a step of the pivot algorithm to proceed with the clustering.

We will show that by combining these two algorithms, it is possible to get an algorithm that is better than each of them individually.
\begin{theorem}\label{thm:alg-approx-ratio}
    \cref{alg:atom-pivot} is a $2.9991$ approximation in time $O(m\log n)$.
\end{theorem}

The reason why such a result is achievable is that the pivot algorithm and the atom-based algorithms have very different worst-case inputs. The pivot algorithm in general performs poorly on input where the input closely resembles disjoint cliques, with a couple of edges either added or removed sparsely. This is exactly the setting where the atom-based algorithm is able to locate atoms.

\begin{algorithm}[H]
\caption{Merge of pivot and atom-finding}\label{alg:atom-pivot}
    \begin{algorithmic}[1]
        \Let{$Q$}{$\emptyset$}\Comment{$Q$ is a multiset}
        \For{$v\in V$}
            \State Add $q(d(v))$ copies of $v$ to $Q$
        \EndFor
        \While{$G$ is not empty}
            \If{$\exists v \in Q$}
                \State Remove $v$ from $Q$.
                \State Try to find a good cluster $K$ near $v$ (\cref{sec:locate-good-cluster})
                \If{a good cluster $K$ is found}
                    \State Construct a cluster $C$ based on $K$ (\cref{sec:pivot-with-good-cluster}), and remove $C$ from $G$
                \EndIf
            \Else
                \State Construct a cluster $C$ by a normal pivot step (\cref{sec:pivot-without-good-cluster}), and remove $C$ from $G$
            \EndIf
            \If{a cluster $C$ was removed}
                \For{$e = (u,v)\in E$ with $u\in C, v\not\in C$}
                    \State Add $p(d(v))$ copies of $v$ to $Q$
                \EndFor
            \EndIf
        \EndWhile
    \end{algorithmic}
\end{algorithm}

The analysis we will follow will be based on the fact that after each operation we do, we can bound the ratio between the cluster we remove and the decrease of cost in the optimal solution after removing the cluster in expectation. This is the standard method that is typically used to analyse a lot of clustering algorithms.

For this result, we need two notions of goodness for clusters, generally based upon how close they are to cliques. Informally, we say a cluster $C$ is $\eps$-\textit{good} if every vertex in the cluster is incident to at most $\eps|C|$ errors. Similarly to this, we call a cluster $\eps$-\textit{good-on-average} if the whole cluster is incident to at most an $\eps|C|^2$ errors. 

To get the main result there are three main technical ingredients. In \cref{sec:pivot-without-good-cluster}, we show that if you run a single step of the pivot algorithm on a graph, where you are given the guarantee that no $\eps$-good-on-average clusters on average exists, then this step will in expectation be strictly better than $3$ approximate:

\begin{restatable*}{theorem}{pivotnogood}
    \label{thm:pivot-without-good-clusters}
    If the optimal clustering $\mathcal{C}^*$ contains no $\eps^2$-good-on-average clusters then for step $t$ if we create a cluster $C$ from a pivot step then,
    $$\E[\alg_t] \le \left(3-\frac{\eps^2}{\frac{5}{6}\eps^2 + 1}\right)\E[\opt_t]$$
    and this can be done in time $O\left(\sum_{v\in C}d(v)\right)$.
\end{restatable*}

This immediately motivates the next results. To use this result, we need to guarantee that at any time when we want to perform a pivot step, no $\eps$-good-on-average clusters are present in the graph. First we show that we can locate each $\eps$-good-on-average cluster, for sufficiently small $\eps$, with high probability under edge and vertex deletions in the graph, in total time $O(m\log n)$.

\begin{restatable*}{theorem}{locatinggoodclusters}
\label{thm:locating-good-clusters}
    In expected $O(m\log n)$ total time under edge and vertex deletions, we can report good clusters such that:
    \begin{itemize}
        \item The reported clusters are always $\frac{4\eps+3\delta+2\delta^2+\eps\delta}{1-3\eps-2\gamma(1-\eps)(1+\delta)}$-good;
        \item $\eps^2$-good-on-average clusters are never missed with high probability.
    \end{itemize}
\end{restatable*}

The final result is then on how to handle these atoms. It is not feasible to just remove them, since this has the issue that it can introduce a lot of cost to the clustering. We however show that if you for every vertex that has sufficiently many neighbours inside the cluster compared to outside the cluster make a choice at random based on this ratio, then in expectation the removed cluster will be $2+O(\eps)$-approximate, specifically:
\begin{restatable*}{theorem}{atomrounding}
\label{thm:pivot-with-good-cluster}
    Given an $\eps'$-good cluster $K$ where $\eps'<\frac16$, we can construct a cluster $C$, such that when removing $C$ from the graph at step $t$ we have
    $$\E[\alg_t] \le \left(2+\frac{7\eps'(1+2\eps')}{2(1-2\eps')^2}\right) \E[\opt_t]$$
    and this can be done in time $O\left(\sum_{v\in C}d(v)\right)$.
\end{restatable*}

This result can be shown by splitting into three cases and analysing each of them. For a given vertex incident to an atom, if it has way more neighbours inside the atom cluster compared to outside the atom cluster then we will never make too many mistakes compared to opt if we put it into the cluster. Otherwise, if there are way more neighbours outside the atom cluster compared to inside the atom cluster then we will never make many errors by not including the vertex in the cluster. Finally

We believe that this procedure and result could be of independent interest, since it gives a clean and practical method to preprocess the graph so that no atoms are present in the graph, while also being able to guarantee that the cost of the optimal solution incident to the removed.

Finally, if we fix the parameter $\eps$ such that, the value for which we can guarantee that there are no $\eps^2$ good-on-average clusters in the graph. 

Whenever we construct a step of the algorithm, we have the two cases of \cref{thm:locating-good-clusters}. Either the removed cluster was an atom, in which case \cref{thm:pivot-with-good-cluster} can be used. Otherwise, we get the setting of \cref{thm:pivot-without-good-clusters}. By balancing these two, we can get the main result, as follows.

\begin{proof}[Proof of \cref{thm:alg-approx-ratio}]
    By \cref{thm:locating-good-clusters}, we are in one of two cases. Either no $\eps^2$-good-on-average exists in the graph with high probability, or we locate a $\frac{4\eps+3\delta+2\delta^2+\eps\delta}{1-3\eps-2\gamma(1-\eps)(1+\delta)}$-good cluster. If no $\eps^2$-on-average exists, then 
    $\alg_t \le \left(3-\frac{\eps^2}{\frac{5}{6}\eps^2 + 1}\right)\opt_t$
    by \cref{thm:pivot-without-good-clusters}. If on the other hand a $\eps'=\frac{4\eps+3\delta+2\delta^2+\eps\delta}{1-3\eps-2\gamma(1-\eps)(1+\delta)}$-good cluster is located, then by \cref{thm:pivot-with-good-cluster} we have
    $\alg_t \le \left(2+\frac{7\eps'(1+2\eps')}{2(1-2\eps')^2}\right) \opt_t.$
    By selecting $\delta=\gamma=\eps/100$ and $\eps=0.0287$ we get 
    $\alg_t \le 2.9991\opt_t.$
    Then by applying \cref{lem:expected-step} the statement follows.
\end{proof}

\subsection{Structure of the paper}
The paper will be structured as follows. We will state some preliminaries in \cref{sec:preliminaries}. Then we will analyse the pivot step in graphs without good clusters in \cref{sec:pivot-without-good-cluster}. Then we will describe how to locate good clusters dynamically in \cref{sec:locate-good-cluster}, and how to remove them in \cref{sec:pivot-with-good-cluster}. 
Finally, we will show the practicality of the algorithm through a series of experiments in \cref{sec:experiments}.

\section{Preliminaries}\label{sec:preliminaries}
We work on an undirected unweighted graph $G=(V,E)$, where vertices and edges are being removed during the pivot process.
Let $n$ and $m$ denote the number of vertices and edges \emph{in the beginning}, respectively.
For a vertex $v$, we use $N(v)$ to denote the set of vertices that are neighbours to $v$ \emph{currently}, including $v$ itself, and let $d(v)=|N(v)|$ denote the \emph{current} degree\footnote{Actually $d(v)$ is the degree plus $1$. However, for simplicity, in this paper we refer to $d(v)$ when we say the degree of $v$.} of $v$.
For two sets $A$ and $B$, we use $A\triangle B$ to denote the symmetric difference between them.

A \emph{cluster} is a set of vertices $C\subseteq V$. A \emph{clustering} of $G$ is a set of clusters that form a partition of $V$.
The cost of a clustering is the number of edges across different clusters plus the number of pairs of vertices in the same clustering without sharing an edge.
The local goodness of a cluster is defined in the following two different ways.

\begin{definition}\label{def:good}
    A cluster $C$ is $\eps$-good if $|N(u)\triangle C| \le \eps|C|$ for any $u\in C$.
\end{definition}

\begin{definition}\label{def:avg-good}
    A cluster $C$ is $\eps$-good-on-average if $$\sum_{u\in C}|N(u)\triangle C| \le \eps|C|^2.$$
\end{definition}

Previous analysis of various clustering algorithms has relied on the approach that we split the process of the algorithm into separate steps. For each of these steps, given some graph $G$ we then construct a new graph $G'$ by removing some cluster $C$ and all incident edges. We then show that the expected cost of the cluster $C$ is at most some multiplicative constant $\alpha$ larger than the decrease in the optimal solution between $G$ and $G'$. 

We are going to show this by using similar triangle counting techniques to previous pivot based clustering algorithms, though arguing that we must have a large number of specific triangles in the graph. The central statement we will be showing is the following. Specifically, we will base our analysis on the works of \cite{ACN08,CMSY15}.
From their work we have the following lemma central to the analysis.

\begin{lemma}[\cite{CMSY15}]\label{lem:expected-step}
    Consider some pivot based algorithm $\alg$, and let $\alg_t$ be the number of constraints violated at step $t$. Let $\opt_t$ be the decrease of the cost of the optimal solution at step $t$. Then if $\E[\alg_t] \le \alpha \E[\opt_t]$ then $\alg$ is an $\alpha$ approximation.
\end{lemma}

This means that the analysis in this paper will focus on bounding $\E[\opt_t]$ and $\E[\alg_t]$ in expectation, since this together with the lemma implies an approximation algorithm. 

\section{The pivot in the absence of good clusters}
\label{sec:pivot-without-good-cluster}
In this section we are going to prove one of our central results, namely that if you run the regular pivot algorithm in a graph where no good clusters exist a single step, then this step contributes to the approximation ratio being strictly less than $3$.

This is the case, even though it has been shown that for general graphs the approximation ratio of the pivot algorithm is tight. Nonetheless, we show that in this specific graph family this is not the case. The reason for this is that all the examples that actually forces the approximation ratio of the pivot to be three necessarily must contain some very good clusters, which are exactly the onces we exclude.

With $\opt_t$ we here imply the difference between the optimal solution of the graph before doing a pivot step, and after doing a pivot step. This could in principle change the optimal clustering. However, after doing modifications, the previous clustering is still a valid upper bound for the optimal solution. So unless explicitly stated, we are not going to charge any cost to the change in the optimal solution after removing a cluster.

\begin{figure}[ht]
    \centering
    \begin{center}
        \begin{tabular}{|c|c|c|c|c|c|c|} \hline
             & Edges$(u,v,w)$ & $\optv(u,v,w)$ & $g(u,v,w)$ & $cost(u,v,w)$ & Illustration\\\hline
             \multirow{4}{*}{$u,v,w\in C$}
             &$+++$ & \multicolumn{2}{c|}{0} & 0 & \multirow{4}{*}{\begin{tikzpicture}[scale=1]
  \draw (0,0) circle (0.75);

  \coordinate (u) at (-0.25,0.5);   %
  \coordinate (v) at (0.5,0.13);    %
  \coordinate (w) at (-0.25,-0.6);   %

  \draw[thick]
    (u) -- node[midway, above] {$1$} (v)
    -- node[midway, right] {$2$} (w)
    -- node[midway, left] {$3$} (u);

  \fill (u) circle (2pt) node[below left] {};
  \fill (v) circle (2pt) node[left] {};
  \fill (w) circle (2pt) node[right] {};
\end{tikzpicture}}\\
             &$++-$ & \multicolumn{2}{c|}{1} & 3 & \\\cline{3-4}
             &$+--$ & 2 & 1 & 0 & \\

             &$---$ & 0 & 3 & 0 & \\\hline
             \multirow{6}{*}{$u,v\in C, \;w\in V\setminus C$}
             &$+++$ & \multicolumn{2}{c|}{2} & 0 & \multirow{6}{*}{\begin{tikzpicture}[scale=1]
  \draw (0,0) circle (0.75);

  \coordinate (A) at (0.25,-0.4);   %
  \coordinate (B) at (0.25,0.4);   %
  \coordinate (C) at (1.2,0.15);    %

  \draw[thick]
    (A) -- node[midway, left] {$1$} (B)
    -- node[midway, above] {$2$} (C)
    -- node[midway, below right] {$3$} (A);

  \fill (A) circle (2pt);
  \fill (B) circle (2pt);
  \fill (C) circle (2pt);
\end{tikzpicture}}\\
             &$++-$ & \multicolumn{2}{c|}{1} & 3 & \\
             &$+--$ & \multicolumn{2}{c|}{0} & 0 & \\
             &$-++$ & \multicolumn{2}{c|}{3} & 3 & \\
             &$-+-$ & \multicolumn{2}{c|}{1} & 0 & \\
             &$---$ & \multicolumn{2}{c|}{0} & 0 & \\\hline
             \multirow{6}{*}{$u\in C,\;v,w\in V\setminus C$}
             &$+++$ & \multicolumn{2}{c|}{2} & 0 & \multirow{6}{*}{\begin{tikzpicture}[scale=1]
  \draw (0,0) circle (0.75);

  \coordinate (A) at (0.35,0.1);   %
  \coordinate (B) at (1.2,0.55);    %
  \coordinate (C) at (1.25,-0.4);   %

  \draw[thick]
    (A) -- node[midway, above] {$1$} (B)
    -- node[midway, right] {$2$} (C)
    -- node[midway, below] {$3$} (A);

  \fill (A) circle (2pt);
  \fill (B) circle (2pt);
  \fill (C) circle (2pt);
\end{tikzpicture}}\\
             &$++-$ & \multicolumn{2}{c|}{1} & 3 & \\
             &$+--$ & \multicolumn{2}{c|}{0} & 0 & \\
             &$+-+$ & \multicolumn{2}{c|}{2} & 3 & \\
             &$-+-$ & \multicolumn{2}{c|}{0} & 0 & \\
             &$---$ & \multicolumn{2}{c|}{0} & 0 & \\\hline
        \end{tabular}
    \end{center}
    \caption{Decrease of the cost of the optimal solution $\optv$, an amortisation $g$ of this cost taking changing the optimal clustering into account, and the cost of the pivot algorithm $\cost$ all up to symmetries in the labelling of triangles.}
    \label{tab:costs-and-amortisation}
\end{figure}

We are going to define $\optv(u,v,w)$ and $\cost(u,v,w)$ similarly to the original pivot algorithm \cite{ACN08}. The idea is that it should be a lower bound on the amount of cost that the optimal solution is lowered by in expectation at each pivot step. 

\begin{lemma}[\cite{ACN08}]\label{lem:cost-ge-alg}
    Let $\cost(u,v,w)$ be defined as in \cref{tab:costs-and-amortisation}. Then 
    
    $$\frac{1}{6|V|} \sum_{u,v,w\in V}\cost(u,v,w) \ge \E[\alg_t] $$
\end{lemma}

As a note, to make the counting simpler, whenever we sum over multiple elements of a set, as in the expression $\sum_{u,v,w\in C} g(u,v,w)$, this will be counting different permutations of elements as different summands. That is each triangle will occur 6 times, once for each of its permutations of vertices.

\begin{lemma}\label{lem:g-le-opt}
   Let $g$ be defined as in \cref{tab:costs-and-amortisation}. Then 
    $$\frac{1}{6|V|}\sum_{u,v,w\in V}g(u,v,w)\le \E[\opt_t]$$
\end{lemma}
\begin{proof}
    By subtracting opt on both sides we get
    \begin{equation}\label{eq:triple-minus-bound}
    \frac{1}{6|V|}\sum_{u,v,w\in V}g(u,v,w) -\optv(u,v,w)\le \E[\opt_t]-\frac{1}{6|V|}\sum_{u,v,w\in V}\optv(u,v,w)
    \end{equation}
    This is the statement we will prove. One quantity that is not taken into account when comparing $\opt_t$ to $\frac{1}{6|V|}\sum_{u,v,w\in V}\optv(u,v,w)$ is how changing the optimal clustering after doing a pivot step can lead to a decrease in the cost of the optimal solution. This is what we will bound the left side by. Letting $\mathcal{C}^*$ be an optimal clustering, from the definition of $g$ we get:

    \begin{equation}
    \sum_{u,v,w\in V}g(u,v,w) -\optv(u,v,w) = \sum_{C\in \mathcal{C}^*}\sum_{\substack{u,v,w\in C:\\--\pm\text{ triangle}}}g(u,v,w) -\optv(u,v,w)
    \end{equation}

    Consider some vertex $u\in C\in \mathcal{C}$. Let $c_u^{--+}$ be the number of pairs of vertices $v,w\in C$ such that $u,v,w$ is a $(--+)$ triangle where $u$ is incident to both $-$ edges. Let $c_u^{---}$ be defined similarly for $(---)$ triangles. In addition, we define a function $\Delta_u(v,w)$ such that
    $$\Delta_u(v,w) = \begin{cases}
        -1 & \text{if (u,v,w) is a $(--+)$ with $u$ incident to both $-$ edges} \\
        0 & \text{if (u,v,w) is a $(--+)$ with $u$ not incident to both $-$ edges}\\
        (g(u,v,w)-\optv(u,v,w))/3 & \text{otherwise}
    \end{cases}$$

    From this it can be seen that $g(u,v,w)-\optv(u,v,w) = \Delta_u(v,w)+\Delta_v(u,w)+\Delta_w(u,v)$. So we can write
    \begin{align}
        \sum_{C\in \mathcal{C}}\sum_{\substack{u,v,w\in C:\\--\pm\text{ triangle}}}g(u,v,w) -\optv(u,v,w)&=\sum_{C\in \mathcal{C}}\sum_{u\in C}\sum_{v,w\in C\setminus N(u)}\Delta_u(v,w)\nonumber\\
        &= \sum_{C\in \mathcal{C}}\sum_{u\in C}c_u^{---} - c_u^{--+}\label{eq:mmm-mmp-diff-count}
    \end{align}
    by expanding $\Delta_u(v,w)$ according to its definition, the definition of $g$ and using $\Delta_u$ is non-zero only when $u$ is incident to the two $-$ edges. We can now split into two cases. If $c_u^{--+} \ge c_u^{---}$ then $c_u^{---} - c_u^{--+} \le 0$ and so not changing the optimal cluster covers this since the right hand side of \cref{eq:triple-minus-bound} is non-negative. So we focus on the case $c_u^{--+} \ge c_u^{---}$. 

    If $c_u^{--+} \ge c_u^{---}$, then after selecting $u$ as a pivot, the part of $C$ that remains is exactly $C\setminus N(u)$. So we would get $c_u^{--+}$ internal $+$ edges and $c_u^{---}$ internal $-$ edges from the remaining cluster. This means that by splitting the remaining vertices into singletons, we can create a clustering that decreases by $c_u^{---}-c_u^{--+}$, which implies that the optimal clustering decreases in cost by at least this quantity. Combining this with \cref{eq:mmm-mmp-diff-count}, we get \cref{eq:triple-minus-bound}, which completes the proof.
\end{proof}

Next, we introduce a quantity to count the number of $(++-)$ triangles, which is where the main amount of cost in the pivot algorithm will be coming from. 
\begin{definition}\label{def:ppm-count-ri}
    Fix $\mathcal{C}^*$ as an optimal clustering. Let $r_i^{++-}$ be the number of $(++-)$ triangles such that $i\in\{1,2,3\}$ of the vertices belong to the same cluster in the optimal clustering and the triangle has $cost(u,v,w)=3=3g(u,v,w)$, such that in $r_3^{++-}$ each triangle is counted $6$ times, in $r_2^{++-}$ each triangle is counted $4$ times and in $r_1^{++-}$ each triangle is counted $2$ times for each cluster where it satisfies this condition.

    Finally, let $r^{++-}=r_3^{++-}+r_2^{++-}+r_1^{++-}$
\end{definition}
The way that this should be understood is that if a triangle has two vertices in the same cluster and one in a different one, it will be counted $4$ times in $r_2^{++-}$ and twice in $r_1^{++-}$. It might seem a bit arbitrary with the number of times that each triangle is counted in the different sets, but this is to make sure that when we consider $r^{++-}=r_1^{++-}+r_2^{++-}+r_3^{++-}$ each triangle is counted exactly $6$ times.

To make the proof work, we however have some triangles that we need to count separately, namely the $(++-)$ excluded from being counted by $r^{++-}$, since these have a different ratio compared to the other edges, but still need to be taken into account for the cost of the pivot algorithm.
\begin{definition}
    Fix $\mathcal{C}^*$ as an optimal clustering. Let $s_2^{-++}$ be the number of $(-++)$ triangles with exactly two vertices inside the same optimal cluster, counted with multiplicity $4$. Let $s_1^{+-+}$ be the number of $(+-+)$ with the vertex incident to both $+$ edges in an optimal cluster different from the other vertices.

    Finally, let $a^{++-} = r^{++-} + s_1^{+-+} + s_2^{-++}$
\end{definition}
\begin{lemma}\label{lem:r1r2r3-bound}
    \begin{equation}    
        a^{++-}\le \sum_{C\in C^*}\sum_{u\in C}3|N(u)\setminus C|^2 + 6|C||N(u)\setminus C| + 3|C||C\setminus N(u)|+3|C\setminus N(u)||N(u)\setminus C|
    \end{equation}
\end{lemma}
\begin{proof}
    We are going to bound $r_1^{++-},r_2^{++-},r_3^{++-}, s_1^{+-+}, s_2^{-++}$ in groups, and then later their sum.

    \begin{claim}\label{clm:r1r2r3-bound-r1c}
        $$r_1^{++-} + s_1^{+-+} \le \frac{1}{2}r_2^{++-} + 3\sum_{C\in C^*}\sum_{u\in C}|N(u)\setminus C|^2$$
    \end{claim}
    \begin{proof}
        For each triangle $(u,v,w)$ counted by $r_1^{++-}$ or $s_1^{+-+}$, due to $u$ being in a different cluster compared to $v$ and $w$, then either $(u,v,w)$ is also counted $4$ times by $r_2^{++-}$ if $v$ and $w$ is the same cluster, or $u,v,w$ are all in different clusters, in which case one vertex is incident to two outgoing $+$ edges. If $v$ and $w$ is in the same cluster then $r_1^{++-}$ or $s_1^{+-+}$ only need to count this triangle twice, giving the coefficient of $\frac{1}{2}$. In the case where all $u,v,w$ are in different clusters, the triangle is counted $4$ times by $r_1^{++-}$ and $2$ times in $s_1^{+-+}$ and so this is bounded by $3\sum_{C\in C^*}\sum_{u\in C}|N(u)\setminus C|^2$. 
    \end{proof}
    \begin{claim}\label{clm:r1r2r3-bound-r2c}
        $$r_2^{++-} + s_2^{-++} \le \sum_{C\in C^*}\sum_{u\in C}4|C||N(u)\setminus C| + 2|C\setminus N(u)||N(u)\setminus C|$$
    \end{claim}
    \begin{proof}
        For each $(++-)$ triangle $(u,v,w)$ counted by $r_2^{++-}$ due $u,v$ being inside some optimal cluster $C$, exactly one of $u$ and $v$ is incident to both $+$ edges. The number of such pairs for is then counted by $\sum_{C\in C^*}\sum_{u\in C}|C||N(u)\setminus C|$. Each triangle is then counted at least once and so since $r_2^{++-}$ counts each triangle $4$ times, the term has the coefficient $4$. For each $(-++)$ triangle, either $u$ or $v$ is incident to an internal $-$ edge and an external $+$ edge. The number of such pairs are $\sum_{C\in C^*}\sum_{u\in C}|C\setminus N(u)||N(u)\setminus C|$, where each $(+-+)$ triangle is counted twice, giving the coefficient of $2$.
    \end{proof}
    \begin{claim}\label{clm:r1r2r3-bound-r3c}
        $$r_3^{++-} \le 3\sum_{C\in C^*}\sum_{u\in C}|C||C\setminus N(u)|$$
    \end{claim}
    \begin{proof}
        For each $(++-)$ triangle counted by $r_3^{++-}$, two vertices are incident to both a $-$ edge and a $+$ edge. This gives that for the expression $\sum_{C\in C^*}\sum_{u\in C}|C||C\setminus N(u)|$, each triangle $(++-)$ is twice. Since $r_3^{++-}$ counts each triangle $6$ times the claim follows.
    \end{proof}

    By now combining \cref{clm:r1r2r3-bound-r1c,clm:r1r2r3-bound-r2c,clm:r1r2r3-bound-r3c}, and using that $s_2^{-++}$ when applying \cref{clm:r1r2r3-bound-r2c} is non-negative, we complete the proof.
\end{proof}

\begin{lemma}\label{lem:g-ge-r3}
    Let $\mathcal{C}^*$ be an optimal clustering. Then
    \begin{equation}\label{eq:r3-extra-counting}
        \sum_{C\in\mathcal{C}^*}\sum_{u,v,w\in C}g(u,v,w)\ge r_3^{++-} + 6\sum_{C\in\mathcal{C}^*}\sum_{u\in C}|C\setminus N(u)|^2
    \end{equation}
\end{lemma}
\begin{proof}
    By counting, every $++-$ completely internal in an optimal cluster contributes $g(u,v,w)=1$ to the lhs of \cref{eq:r3-extra-counting}. This is the quantity $r_3^{++-}$. For the expression $\sum_{u\in C}|C\setminus N(u)|^2$, this counts the number of pairwise incident $-$ edges to vertices internal in $C$, and so every $(--+)$ triangle is counted once, each contributing $g(u,v,w)=1$ to the lhs, and every $(---)$ is counted 3 times, but it contributed $g(u,v,w)=3$ to the lhs. Finally, since the lhs counts each triangle $6$ times, we get that the contribution of $(---)$ and $(--+)$ triangles is $\sum_{C\in\mathcal{C}^*}\sum_{u\in C}6|C\setminus N(u)|^2$.
\end{proof}

\begin{lemma}\label{lem:g-ge-r2}
    Let $\mathcal{C}^*$ be an optimal clustering. Then
    \begin{equation}\label{eq:r2-extra-counting}
        2\sum_{C\in\mathcal{C}^*}\sum_{u,v\in C, w\in V\setminus C}g(u,v,w)\ge r_2^{++-} + s_2^{-++}+4\sum_{C\in\mathcal{C}^*}\sum_{u\in C}|C\setminus N(u)||N(u)\setminus C|
    \end{equation}
\end{lemma}
\begin{proof}
    By counting, each $(++-)$ triangle with two vertices in the same optimal cluster contributes $g(u,v,w)=1$ to the lhs of \cref{eq:r2-extra-counting}, which gives the term $r_2^{++-}$.
    For the expression $\sum_{C\in\mathcal{C}^*}\sum_{u\in C}|C\setminus N(u)||N(u)\setminus C|$, this counts the number of pairs of $+$ and $-$ edges such that the $+$ is outside $C$ and the $-$ is inside $C$. Each one of these is either a $(-+-)$ or a $(-++)$ triangle (recall that the ordering of edges matter, see \cref{tab:costs-and-amortisation}). Each $(-+-)$ triangle is counted once by this expression and contributes $g(u,v,w)=1$ to the lhs, while each $(-++)$ triangle is counted twice, and contributes $g(u,v,w)=3$ to the lhs, giving an additional $s_2^{-++}$. Finally, since each triangle on the left is counted $4$ times, we get that the contribution of $(-+-)$ and $(-++)$ triangles is $4\sum_{C\in\mathcal{C}^*}\sum_{u\in C}|C\setminus N(u)||N(u)\setminus C|$.
\end{proof}

\begin{lemma}\label{lem:g-ge-r1}
    Let $\mathcal{C}^*$ be an optimal clustering. Then
    \begin{equation}\label{eq:r1-extra-counting}
        \sum_{C\in\mathcal{C}^*}\sum_{u\in C, v, w\in V\setminus C}g(u,v,w)\ge r_1^{++-} + s_1^{+-+} + 2\sum_{C\in\mathcal{C}^*}\sum_{u\in C}|N(u)\setminus C|^2
    \end{equation}
\end{lemma}
\begin{proof}
    By counting, each $(++-)$ triangle with one vertices in the same optimal cluster contributes $g(u,v,w)=1$ to the lhs of \cref{eq:r1-extra-counting}, which gives the term $r_1^{++-}$.
    For the expression $\sum_{C\in\mathcal{C}^*}\sum_{u\in C}|N(u)\setminus C|^2$, this counts the number of pairs of $+$ edges such that the $+$s is outside $C$. Each one of these is either a $(+++)$ or a $(+-+)$ triangle (Recall that the ordering matters, see \cref{tab:costs-and-amortisation}). Each $(+++)$ and $(+-+)$ triangle is then counted twice per orientation that has a vertex alone in an optimal cluster. Finally, to see the $s_1^{+-+}$ term, observe that for every $(+-+)$ triangle, one of two cases occurred. Either the other two vertices was in different clusters, in which case the triangle was never counted, or the two vertices was in the same cluster, in which case the contribution of the triangle was actually $3$ and not $2$. In both cases, we cover the number of such triangles.
\end{proof}

\begin{lemma}\label{lem:cost-to-ri}
    $$\sum_{u,v,w\in V}\cost(u,v,w)\le 3a^{++-}$$
\end{lemma}
\begin{proof}
    From \cref{tab:costs-and-amortisation}, the only triangles that are not counted by $r^{++-}=r_1^{++-}+r_2^{++-}+r_3^{++-}$ according to \cref{def:ppm-count-ri} are the $(-++)$ triangles with two vertices inside an optimal cluster and $(+-+)$ triangles with one vertex inside an optimal cluster. These are actually the same triangles, and so can be bounded by $3\sum_{C\in \mathcal{C}^*}\sum_{u\in C}|N(u)\setminus C|^2$ to count them a total of $6$ times.
\end{proof}

\begin{lemma}\label{lem:avg-sqr}
    Let $a_1,\ldots,a_n \in \mathbb{R}$. Then
    $$\sum_{i=1}^n a_i^2 \ge \frac{1}{n}\left(\sum_{i=1}^na_i\right)^2$$
\end{lemma}
\begin{proof}
    \begin{align*}
        \sum_{i=1}^n a_i^2 - \frac{1}{n}\left(\sum_{i=1}^na_i\right)^2 &= \frac{1}{n}\left((n+1)\sum_{i=1}^na_i^2 - 2\sum_{i<j}a_ia_j\right)\\
        &= \frac{1}{n}\sum_{i<j}(a_i-a_j)^2 \ge 0
    \end{align*}
\end{proof}

\pivotnogood
\begin{proof}
    First, the running time is trivial. To show the approximation ratio, we are going to compute an upper-bound on the ratio $\alg_t/\opt_t$. Using \cref{lem:cost-ge-alg,lem:g-le-opt}, we have
    $$\frac{\E[\alg_t]}{\E[\opt_t]}\le \frac{\sum_{u,v,w\in V}\cost(u,v,w)}{\sum_{u,v,w\in V}g(u,v,w)}$$
    Then using \cref{lem:cost-to-ri,lem:g-ge-r3,lem:g-ge-r2,lem:g-ge-r1} we have
    $$\frac{\sum_{u,v,w\in V}\cost(u,v,w)}{\sum_{u,v,w\in V}g(u,v,w)}\le \frac{3a^{++-} }{a^{++-} + \sum_{C\in\mathcal{C}^*}\sum_{u\in C}6|C\setminus N(u)|^2+4|C\setminus N(u)||N(u)\setminus C|+2|N(u)\setminus C|^2}$$

    Since the lhs is at most $3$, increasing $a^{++-}$ will only bring it closer to $3$. So we can apply \cref{lem:r1r2r3-bound} and  rewrite and collect into squares to get 
    \begin{align*}
        \frac{\sum_{u,v,w\in V}\cost(u,v,w)}{\sum_{u,v,w\in V}g(u,v,w)}\le 3-\frac{\sum_{C\in \mathcal{C}}\sum_{u\in C}6(|N(u)\setminus C| + |C\setminus N(u)|)^2}{\sum_{C\in \mathcal{C}}\sum_{u\in C}5(|N(u)\setminus C| + |C\setminus N(u)|)^2 + 6(|N(u)\setminus C| + |C\setminus N(u)|)|C|}.
    \end{align*}
    By applying \cref{lem:avg-sqr} and grouping based on the optimal clusters, where $|N(u)\setminus C| + |C\setminus N(u)| = |N(u)\triangle C|$ we finally have
    $$\frac{\sum_{u,v,w\in V}\cost(u,v,w)}{\sum_{u,v,w\in V}g(u,v,w)}\le 3-\frac{\sum_{C\in \mathcal{C}}\frac{6}{|C|}\left(\sum_{u\in C}|N(u)\triangle C|\right)^2}{\sum_{C\in \mathcal{C}}\frac{5}{|C|}\left(\sum_{u\in C}|N(u)\triangle C|\right)^2 + 6(\sum_{u\in C}|N(u)\triangle C|)|C|}.$$
    By now comparing per cluster, we have for every cluster $C\in \mathcal{C}^*$ that
    \begin{align*}
       \frac{\frac{6}{|C|}\left(\sum_{u\in C}|N(u)\triangle C|\right)^2}{\frac{5}{|C|}\left(\sum_{u\in C}|N(u)\triangle C|\right)^2 + 6(\sum_{u\in C}|N(u)\triangle C|)|C|} &= \frac{6\left(\sum_{u\in C}|N(u)\triangle C|\right)}{5\left(\sum_{u\in C}|N(u)\triangle C|\right) + 6|C|^2}\\
       &\ge \frac{6\eps^2}{5\eps^2+ 6}
    \end{align*}
    with the final inequality using that there are no $\eps^2$-good-on-average clusters and the definition \cref{def:avg-good}. From this, we can finally conclude that 
    \begin{align*}
        \frac{\sum_{u,v,w\in V}\cost(u,v,w)}{\sum_{u,v,w\in V}g(u,v,w)}
        &\le 3-\frac{\eps^2}{\frac{5}{6}\eps^2 + 1}
    \end{align*}
    $$$$
\end{proof}

\section{Locating very good clusters}
\label{sec:locate-good-cluster}

In this section, we are going to show how to locate all very good clusters as the vertices are deleted.
We model the task as follows. There is a graph $G$ and an adaptive adversary. In each round, our algorithm needs to decide whether there is a good cluster, then the adversary removes one or more vertices. In addition, if the answer is yes, the algorithm should return a good cluster $K$, and the adversary must remove all vertices from $K$.
The algorithm should satisfy the following properties in each round (with high probability):
\begin{itemize}
    \item If it says yes, the returned cluster must be ``good enough'';
    \item If it says no, there must not exist ``very good'' clusters.
\end{itemize}

First, let us deal with all initial good clusters by a simple reduction. We assume that every vertex $v$ has $\Theta(d(v))$ unique neighbors that were deleted in the 0-th round, i.e., before the algorithm starts.
With this assumption, intuitively speaking, good clusters can only appear during vertex deletions. It will be formally used when we prove \cref{lem:hit-good-cluster}.

No matter which goodness definition (\cref{def:good,def:avg-good}) we use, it is straightforward to see that any cluster will not get worse when other vertices are deleted. Furthermore, a cluster can only get better when some of its neighbors are deleted.
So, after deleting a vertex $u$, we only need to check at each neighbor $v\in N(u)$ whether there is a good cluster containing $v$ (and locate it if there is).
To do this, our basic idea is to run the following $\clean$ procedure for the neighborhood $N(v)$, which was introduced by \cite{cao2024understanding} as a preprocessing step. We note that $\alpha$ and $\beta$ are small constants that will be determined later.

\begin{algorithm}[H]
\caption{$\clean[\alpha,\beta](C)$}
    \begin{algorithmic}[1]
        \State $K\gets\{u\in C \mid |N(u)\triangle C| \le \alpha|C|\}$
        \If{$|K| \ge (1-\beta)|C|$}
            \State \Return $K$
        \Else
            \State \Return $\emptyset$
        \EndIf
    \end{algorithmic}
\end{algorithm}

We first show that the cluster returned by the $\clean$ procedure is always ``good enough''.

\begin{lemma}
\label{lem:clean}
    Let $K=\clean[\alpha,\beta](C)$. If $K\neq\emptyset$, then $K$ is $\frac{\alpha+\beta}{1-\beta}$-good.
\end{lemma}

\begin{proof}
    Since $|K|\ge(1-\beta)|C|$ and $K\subseteq C$, we have $|C\triangle K|\le \beta|C|$.
    Then, for any $v\in K$, we have $|N(v)\triangle K|\le|N(v)\triangle C|+|C\triangle K|\le (\alpha+\beta)|C| \le \frac{\alpha+\beta}{1-\beta}|K|$.
\end{proof}

However, it will take $O(d^2(v))$ time to run the $\clean$ procedure, so we don't want to do it every time. The next lemma shows that, to make sure that we can find all ``very good'' clusters, we actually have many chances. We note that $\eps$ and $\delta$ are small constants that will be determined later. The proof is deferred to \cref{app:hit-good-cluster}.

\begin{definition}
    A vertex $v$ is $(\alpha,\beta)$-good if $\clean[\alpha,\beta](N(v))\neq\emptyset$.
\end{definition}

\begin{definition}
    A vertex $v$ is $(\alpha,\beta)$-hit when we delete one of its neighbors if it is $(\alpha,\beta)$-good after the deletion.
\end{definition}

\begin{lemma}
\label{lem:hit-good-cluster}
    When there is an $\eps^2$-good-on-average cluster $C$, there exists a vertex $v\in C$ such that $v$ has been $\left(\frac{2\eps+2\delta}{1-\eps},\frac{\frac{2\eps}{1-\eps}+\delta}{1+\delta}\right)$-hit $\lceil\delta d(v)\rceil$ times when its degree\footnote{When we talk about degrees when a deletion occurs, we always refer to degrees right after the deletion.} was less than $(1+\delta)d(v)$.
    Furthermore, $v$ was always being $\left(\frac{2\eps+2\delta}{1-\eps},\frac{\frac{2\eps}{1-\eps}+\delta}{1+\delta}\right)$-good since the first of these hits.
\end{lemma}

As a vertex needs to be hit $\Omega(d(v))$ times for a ``very good'' cluster to appear, we don't need to guarantee that we can deterministically find a ``good enough'' cluster when hitting it. Instead, we can afford to miss it with probability $e^{-\Omega(\log n/d(v))}$ for each hit.
So, before running the $\clean$ procedure, we first estimate whether it will succeed by sampling vertices from $N(v)$ and estimating $|N(u)\triangle N(v)|$ for each sample $u$.
Only when the following $\check$ procedure returns True, we actually run the $\clean$ procedure.

\begin{algorithm}[H]
\caption{$\check[\alpha,\beta,\gamma](v)$}
    \begin{algorithmic}[1]
        \State $k \gets 0$
        \State $\eta \gets \frac{16}{\beta\gamma^2}\ln d(v)$
        \For{$i=1,2,\dots,\eta$}
            \State $u_i \gets$ a uniformly random sample from $N(v)$
            \State $k' \gets 0$
            \State $\eta' \gets \frac{8}{\alpha\gamma^2}\ln\frac{2}{\beta\gamma}$
            \For{$j=1,2,\dots,\eta'$}
                \State $w_j \gets$ a uniformly random sample from $N(v)$
                \If{$w_j\notin N(u_i)$}
                    \State $k' \gets k'+1$
                \EndIf
            \EndFor
            \If{$k'>\frac{1}{2}\left(1+\alpha(1-\gamma)-\frac{d(u_i)}{d(v)}\right)\eta'$}
                \State $k \gets k+1$
            \EndIf
        \EndFor
        \If{$k>\beta(1-\gamma)\eta$}
            \State \Return False
        \Else
            \State \Return True
        \EndIf
    \end{algorithmic}
\end{algorithm}

We note that $\gamma$ is a small constant that will be determined later.
The next two lemmas bound the error probability of the $\check$ procedure. The proofs are included in \cref{app:check-false,app:check-true}.

\begin{lemma}
\label{lem:check-false}
    If $\clean[\alpha,\beta](N(v))=\emptyset$, then $\check[\alpha,\beta,\gamma](v)=\text{False}$ with probability at least $1-1/d^2(v)$.
\end{lemma}

\begin{lemma}
\label{lem:check-true}
    If $\clean[\alpha(1-2\gamma),\beta(1-2\gamma)](N(v))\neq\emptyset$, then $\check[\alpha,\beta,\gamma](v)=\text{True}$ with probability at least $1-1/d^2(v)$.
\end{lemma}

Our full algorithm for locating good clusters is as follows.

\begin{algorithm}[H]
\caption{$\atomfinding[\eps,\delta,\gamma](G=(V,E))$}
\label{alg:atom-finding}
    \begin{algorithmic}[1]
        \Let{$n$}{$|V|$}
        \Let{$\alpha$}{$\frac{2\eps+2\delta}{(1-\eps)(1-2\gamma)}$}
        \Let{$\beta$}{$\frac{\frac{2\eps}{1-\eps}+\delta}{(1+\delta)(1-2\gamma)}$}
        \Let{$Q$}{$\emptyset$}\Comment{$Q$ is a multiset}
        \For{$v \in V$}
            \State Add $\left\lceil\frac{2(1+\delta)\ln n}{\ln\frac{d(v)}{1+\delta}}\right\rceil$ copies of $v$ to $Q$.
        \EndFor
        \While{$G$ is not empty}
            \If{$\exists v \in Q$}
                \State Remove $v$ from $Q$.
                \If{$v\in V$ and $\check[\alpha,\beta,\gamma](v)=\text{True}$}
                    \State $K\gets\clean[\alpha,\beta](N(v))$
                    \If{$K\neq\emptyset$}
                        \State Report $K$.
                    \EndIf
                \EndIf
            \Else
                \State Report that there are no good clusters.
            \EndIf
            \If{the adversary decides to remove a cluster $C$}
                \For{$u\in C$}
                    \For{$v\in N(u)$}
                        \State Remove $(u,v)$ from $E$.
                        \State Add $\left\lceil\frac{2(1+\delta)\ln n}{\delta d(v)\ln\frac{d(v)}{1+\delta}}\right\rceil$ copies of $v$ to $Q$.
                    \EndFor
                    \State Remove $u$ from $V$.
                \EndFor
            \EndIf
        \EndWhile
    \end{algorithmic}
\end{algorithm}

The next lemma shows that \cref{alg:atom-finding} will only report ``good enough'' clusters.

\begin{lemma}
\label{lem:good-enough}
    The clusters reported by \cref{alg:atom-finding} (Line 13) are always $\frac{4\eps+3\delta+2\delta^2+\eps\delta}{1-3\eps-2\gamma(1-\eps)(1+\delta)}$-good.
\end{lemma}

\begin{proof}
    It simply follows from \cref{lem:clean} and the fact that $\alpha=\frac{2\eps+2\delta}{(1-\eps)(1-2\gamma)}$ and $\beta=\frac{\frac{2\eps}{1-\eps}+\delta}{(1+\delta)(1-2\gamma)}$.
\end{proof}

Recall that $n$ is the number of vertices in the beginning. The next lemma shows that, with high probability, \cref{alg:atom-finding} will never miss ``very good'' clusters. We defer the proof to \cref{app:very-good}.

\begin{lemma}
\label{lem:very-good}
    With probability at least $1-1/n^3$, whenever \cref{alg:atom-finding} reports that there are no good clusters (Line 15), there are no $\eps^2$-good-on-average clusters.
\end{lemma}

Finally, the next lemma shows that \cref{alg:atom-finding} can be implemented in near-linear time. We defer the proof to \cref{app:time-atom-finding}.

\begin{lemma}
\label{lem:time-atom-finding}
    The expected running time of \cref{alg:atom-finding} is $O(m\log n)$.
\end{lemma}

Combining \cref{lem:good-enough,lem:very-good,lem:time-atom-finding}, we conclude this section with the following theorem.

\locatinggoodclusters

\section{Pivots with good clusters}
\label{sec:pivot-with-good-cluster}

Suppose that we know an $\eps'$-good cluster $K$ for some $\eps'<1/6$. In this section, we will show that we are able to locally construct a cluster $C\supseteq K$, such that the expected cost we pay in this step is at most $2+O(\eps')$ times the expected decrease of the optimal cost.

Previous work has shown that such a cluster $K$ is completely inside a cluster in the optimal clustering.

\begin{lemma}[\cite{cao2025static} Lemma 26]
\label{lem:K-in-opt-cluster}
    For $\eps'<1/6$, any $\eps'$-good cluster is completely inside a cluster in the optimal clustering.
\end{lemma}

Furthermore, the cluster can only be slightly larger than $K$. Let $C'$ denote the cluster in the optimal clustering that contains $K$. We have the following lemma.

\begin{lemma}
\label{lem:C'-minus-K}
    $|C'\setminus K| \le 2\eps'|K|$.
\end{lemma}

\begin{proof}
    Since $K$ is $\eps'$-good, the number of edges between $K$ and $|C'\setminus K|$ is at most $\eps'|K|^2$. On the other hand, this number is at least $\frac{1}{2}|K|\cdot|C'\setminus K|$, since otherwise breaking $C'$ into $K$ and $C'\setminus K$ will improve the optimal clustering. So $|C'\setminus K| \le 2\eps'|K|$.
\end{proof}

We first expand $K$ a little to include vertices that are sure to be together with $K$.
For any vertex $v\notin K$, let $\alpha_v=|N(v)\setminus K|/|K|$ and $\beta_v=|N(v)\cap K|/|K|$, i.e., $v$ has $\alpha_v|K|$ neighbors outside $K$ and $\beta_v|K|$ neighbors in $K$.
Whenever there exists some vertex $v$ such that $\beta_v > 1-\beta_v+\alpha_v+2\eps'$, we insert $v$ into $K$. In the end, we will have $2\beta_v \le 1+\alpha_v + 2\eps'$ for all $v\notin K$, and $K$ still has good properties as shown in the next lemma. We defer its proof to \cref{app:atom-expansion}.
We remark that, since we only expand $K$, \cref{lem:C'-minus-K} also remains true after expansion.

\begin{lemma}
\label{lem:atom-expansion}
    After expansion, $K$ is still completely inside $C'$, and the number of edges between $K$ and $V\setminus K$ is still at most $\eps'|K|^2$.
\end{lemma}

Then, we include each vertex $v\in(V\setminus K)$ in $C$ with probability $p_v=\frac{\beta_v}{1+\alpha_v}$ independently. Note that, similar to $C'$, $C$ is also only slightly larger than $K$.

\begin{lemma}
\label{lem:C-minus-K}
    $\E[|C\setminus K|] \le \eps'|K|$.
\end{lemma}

\begin{proof}
    We have
    \begin{equation*}
        \sum_{v\in(V\setminus K)} \beta_v|K| = \sum_{v\in(V\setminus K)}|N(v)\cap K|
        = \sum_{u\in K}|N(u)\setminus K|
        \le \eps'|K|^2,
    \end{equation*}
    where the last inequality comes from \cref{lem:atom-expansion}.
    Then
    \begin{equation*}
        \E[|C\setminus K|] = \sum_{v\in(V\setminus K)} p_v
        \le \sum_{v\in(V\setminus K)} \beta_v
        \le \eps'|K|.
    \end{equation*}
\end{proof}

Let $K^-$ denote the set of pairs of vertices in $K$ that don't share an edge.
The next lemma bounds our expected cost in this step. The proof is basically a lot of computations, which we defer to \cref{app:atom-our-cost}.

\begin{lemma}
\label{lem:atom-our-cost}
    In this step, the expected cost we pay is at most
    $$|K^-| + \sum_{v\in(V\setminus K)} \left( p_v\left(1-\beta_v+\alpha_v+\frac{\eps'}{2}\right)+(1-p_v)\beta_v \right)|K|.$$
\end{lemma}

The next lemma bounds the expected decrease of the optimal cost, whose proof is in \cref{app:atom-opt-cost}.

\begin{lemma}
\label{lem:atom-opt-cost}
    In this step, the cost of the optimal clustering decreases by at least
    $$|K^-| + \sum_{v\in(V\setminus C')}\beta_v|K| + \sum_{v\in(C'\setminus K)}(1-\beta_v+p_v(\alpha_v-2\eps'))|K|$$
    in expectation.
\end{lemma}

Comparing \cref{lem:atom-our-cost,lem:atom-opt-cost}, it only remains to show that, for any vertex $v$,
\begin{equation}
\label{eqn:atom-ratio-1a}
    p_v\left(1-\beta_v+\alpha_v+\frac{\eps'}{2}\right)+(1-p_v)\beta_v \le (2+O(\eps'))\beta_v
\end{equation}
and
\begin{equation}
\label{eqn:atom-ratio-2a}
    p_v\left(1-\beta_v+\alpha_v+\frac{\eps'}{2}\right)+(1-p_v)\beta_v \le (2+O(\eps'))(1-\beta_v+p_v(\alpha_v-2\eps')).
\end{equation}
The above inequalities are satisfied with the ratio $2+\frac{7\eps'(1+2\eps')}{2(1-2\eps')^2}$. We defer the proof to \cref{app:pivot-with-good-cluster}. We also note that the time for constructing $C$ is proportional to the number of edges incident to vertices in $K$. We conclude this section with the following theorem.

\atomrounding

\section{Experiments}\label{sec:experiments}
In this section, we will experiment on graphs synthetic graphs. To generate these graphs we will start with $n=10^3$ vertices. We will then randomly construct a partition of $k$ clusters by uniformly at random assigning each vertex to a cluster, and then adding edges between pairs of vertices in the same cluster. Then for each pair of vertices, we will flip the edge with some small probability $\eps$. We will compare the performance of the traditional pivot algorithm of \cite{ACN08} (\texttt{pivot}), the traditional atom-finding procedure \cite{CLMNP21,DBLP:conf/innovations/Assadi022} (\texttt{atom}) and our newly proposed algorithm (\texttt{atom-pivot}). In addition, we will also compare these to the original clustering induced by the cliques before the added noise as a comparison (\texttt{planted-clique}). 
\begin{figure}[hbt]
    \centering
    \includegraphics[width=0.9\linewidth]{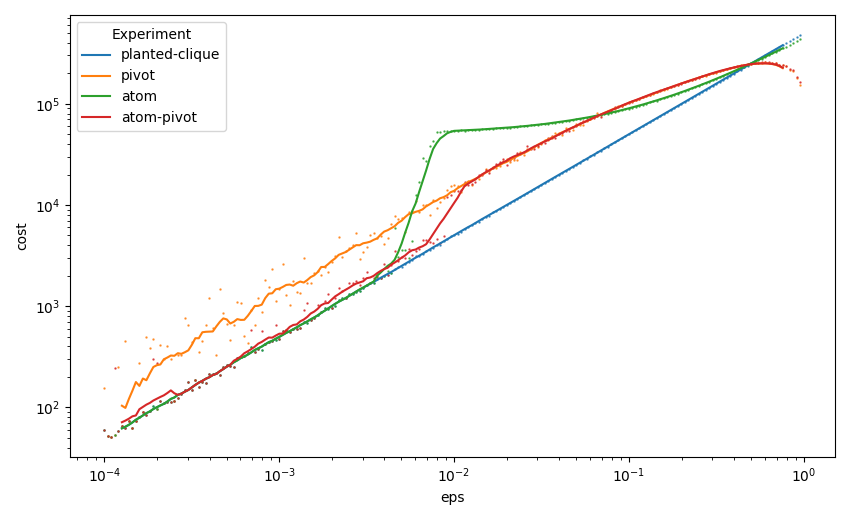}
    \caption{Performance of the algorithms as a function of the noise parameter $\eps$. We generate graphs with $n=10^3$ vertices and a planted partition into $k=10$ clusters. Edges are added between vertices in the same cluster and then independently flipped with probability $\eps$. The y-axis (cost) shows the total number of disagreements in the resulting clustering. For visualization, the plotted curves are smoothed by averaging in log-space over a sliding window of $11$ points: $l_i = \exp\!\left(\frac{1}{11}\sum_{r=i-5}^{i+5} \ln c_r\right)$, where $c_r$ is the observed cost at experiment $r$. Each point corresponds to one of $200$ different values of $\eps$.}
    \label{fig:vary_eps}
\end{figure}

One thing to note is that both the atom-finding procedure and our proposed algorithm are parameterised by the required goodness of atoms. In our experiments, these will be parameterised by the same value with respect to the located atoms.

From these experiments, we see that the proposed \texttt{atom-pivot} algorithm performs roughly equivalently to the atom-finding algorithm for small values of $\eps$, while for large values of $\eps$ it performs as good as the pivot algorithm. This behaviour is expected from the design of the algorithm. When the noise level $\eps$ is small, the planted clusters remain highly structured and large atoms are typically present in the graph. In this regime, the algorithm is able to identify such atoms and effectively behaves like the atom-finding procedure, producing clusterings with very low disagreement cost.

As $\eps$ increases, however, the structure of the planted clustering gradually deteriorates. In particular, the probability that large atoms exist decreases, and the atom-finding procedure begins to fail to locate suitable atoms. This phenomenon is visible in Figure~\ref{fig:vary_eps}, where the cost of the \texttt{atom} algorithm increases sharply around $\eps \approx 10^{-2}$. In contrast, the \texttt{atom-pivot} algorithm naturally transitions to performing pivot steps when no suitable atoms can be identified. Consequently, its performance smoothly approaches that of the classical \texttt{pivot} algorithm as the noise increases.

Another notable feature of the plot is the regime in which the \texttt{atom} algorithm exhibits a sharp degradation in performance. This behaviour can be explained by the fact that the algorithm relies on the presence of sufficiently good atoms; once the noise level surpasses a certain threshold, such atoms no longer exist with high probability, causing the algorithm to return clusterings with significantly larger cost. The \texttt{atom-pivot} algorithm mitigates this issue by falling back to pivoting, thereby avoiding the catastrophic increase in cost observed for the pure atom-based method.

Finally, we observe that the \texttt{planted-clique} baseline provides a natural lower bound on the achievable cost in this experiment. As expected, its cost grows approximately linearly with $\eps$, reflecting the number of corrupted edges introduced by the noise process. For small $\eps$, both the \texttt{atom} and \texttt{atom-pivot} algorithms closely track this lower bound, indicating that they are able to almost perfectly recover the planted clustering. As the noise increases, the gap between the planted solution and the algorithmic outputs grows, reflecting the increasing difficulty of the instance.

Overall, these experiments demonstrate that the proposed \texttt{atom-pivot} algorithm combines the advantages of atom-based methods in low-noise settings with the robustness of pivoting methods in high-noise regimes, achieving consistently good performance across the entire range of $\eps$ values tested.

In addition, we observe a transition in the behaviour of the algorithms around $\eps \approx 10^{-2}$. This transition corresponds to the point where the planted cluster structure becomes sufficiently corrupted that large atoms cease to exist with high probability. When $\eps$ is very small, most edges inside clusters remain positive while most edges between clusters remain negative, and therefore vertices within the same planted cluster share highly consistent neighbourhoods. This structure gives rise to large atoms that can be reliably identified by atom-finding procedures. As $\eps$ increases, however, the expected number of incorrectly labelled edges grows, and the neighbourhoods of vertices within the same cluster become less consistent. Beyond a certain noise level, the structural conditions required for large atoms are unlikely to hold, and atom-based methods lose their advantage. The \texttt{atom-pivot} algorithm is designed precisely for this regime: once suitable atoms are no longer present, the algorithm reverts to pivoting steps, thereby maintaining stable performance even when the planted structure becomes heavily perturbed.

During this transition, it is somewhat expected that \texttt{atom-pivot} performs slightly better than \texttt{atom}. This is likely due to the fact that around this error, it is random whether the planted clusters have degraded sufficiently to not be located by the atom procedures. However, in the case of the \texttt{atom-pivot} the fall-back of using the pivot algorithm is superior to the choice of clustering everything as singletons. This suggests that when there are a variety of both very good and not so good clusters, the \texttt{atom-pivot} performs a sensible clustering for both types of clusters.

\bibliographystyle{alpha}
\bibliography{references}
\appendix
\section{Missing Proofs in \cref{sec:locate-good-cluster}}

\subsection{Proof of \cref{lem:hit-good-cluster}}
\label{app:hit-good-cluster}

Let $K=\{v\in C \mid |N(v)\triangle C| \le \eps|C|\}$.
By \cref{def:avg-good}, $|K|\ge(1-\eps)|C|$.
For any $u,v\in K$, we have
\begin{equation}
\label{eqn:Nv-K}
    |N(v)\triangle K| \le |N(v)\triangle C|+|C\triangle K| \le 2\eps|C|,
\end{equation}
\begin{equation}
\label{eqn:Nu-Nv}
    |N(u)\triangle N(v)| \le |N(u)\triangle C|+|C\triangle N(v)| \le 2\eps|C|,
\end{equation}
\begin{equation}
\label{eqn:C-dv}
    |C| \le \frac{1}{1-\eps}d(v),
\end{equation}
\begin{equation}
\label{eqn:du-dv}
    d(u) \le (1+\eps)|C| \le \frac{1+\eps}{1-\eps}d(v).
\end{equation}

For any $u\in K$, let $t(u)$ denote the time when the degree of $u$ became less than $(1+\delta)d(u)$. (By our assumption that each vertex has lost a constant fraction of neighbors in round $0$, we can guarantee that the degree of $u$ was greater than $(1+\delta)d(u)$ in the beginning.) Let $v\in K$ be the vertex with latest $t(v)$. Recall that the $d(v)$ is the current degree of $v$. So $v$ must experience $\lceil\delta d(v)\rceil$ neighbor deletions when the degree of $u$ was less than $(1+\delta)d(u)$ for all $u\in K$.

It suffices to show that $v$ is $\left(\frac{2\eps+2\delta}{1-\eps},\frac{\frac{2\eps}{1-\eps}+\delta}{1+\delta}\right)$-good since $t(v)$. Consider any moment not earlier than $t(v)$ (and not later than now). Let $d'(u)$ and $N'(u)$ denote the degree and neighborhood of $u$ at that moment, respectively. Note that for each $u\in K$, we have
\begin{equation}
\label{eqn:N'u}
    N(u)\subseteq N'(u),
\end{equation}
\begin{equation}
\label{eqn:d'u}
    d'(u)<(1+\delta) d(u),
\end{equation}
\begin{equation}
\label{eqn:Nu-N'u}
    |N'(u)\setminus N(u)| = d(u)-d'(u) < \delta d(u).
\end{equation}
So, for each $u\in K$, we have
\begin{align*}
    |N'(u)\triangle N'(v)| & \le |N(u)\triangle N(v)| + |N'(u)\setminus N(u)| + |N'(v)\setminus N(v)| \\
    & \le 2\eps|C| + \delta d(u) + \delta d(v)  \tag{\cref{eqn:Nu-Nv,eqn:Nu-N'u}} \\
    & \le \frac{2\eps+2\delta}{1-\eps} d(v)  \tag{\cref{eqn:C-dv,eqn:du-dv}} \\
    & \le \frac{2\eps+2\delta}{1-\eps} d'(v).
\end{align*}
On the other hand,
\begin{align*}
    |N'(v) \cap K| & \ge |N(v) \cap K|  \tag{\cref{eqn:N'u}} \\
    & \ge |N(v)|-|N(v)\triangle K| \\
    & \ge d(v)-2\eps|C|  \tag{\cref{eqn:Nv-K}} \\
    & \ge \left(1-\frac{2\eps}{1-\eps}\right)d(v)  \tag{\cref{eqn:C-dv}} \\
    & \ge \frac{1-\frac{2\eps}{1-\eps}}{1+\delta}d'(v)  \tag{\cref{eqn:d'u}} \\
    & = \left(1-\frac{\frac{2\eps}{1-\eps}+\delta}{1+\delta}\right)d'(v).
\end{align*}
So $v$ is $\left(\frac{2\eps+2\delta}{1-\eps},\frac{\frac{2\eps}{1-\eps}+\delta}{1+\delta}\right)$-good.

\subsection{Proof of \cref{lem:check-false}}
\label{app:check-false}

For a sample $u_i\in N(v)$, $k'$ is the sum of $\eta'$ independent Bernoulli random variables, each with probability $\frac{|N(v)\setminus N(u_i)|}{d(v)}$ being $1$. So,
\begin{align*}
    \E[k'] & = \frac{|N(v)\setminus N(u_i)|}{d(v)}\eta' \\
    & = \frac12\frac{d(v)-d(u_i)+|N(v)\triangle N(u_i)|}{d(v)}\eta'.
\end{align*}
Assume that $|N(v)\triangle N(u_i)|=\alpha'd(v)$ for some $\alpha'\ge\alpha$, then
\begin{equation*}
    \E[k'] = \frac12\left(1-\frac{d(u_i)}{d(v)}+\alpha'\right)\eta'.
\end{equation*}
Let $$t=\frac12\left(1-\frac{d(u_i)}{d(v)}+\alpha(1-\gamma)\right)\eta' < \E[k'].$$
We are going to bound $\Pr[k'\le t]$.
Without loss of generality, assume $t\ge 0$.
Then, by the Chernoff bound,
\begin{align*}
    \Pr[k' \le t] & \le \exp\left(-\frac{(\E[k']-t)^2}{2\E[k']}\right) \\
    & = \exp\left( -\frac{(\alpha'-\alpha(1-\gamma))^2\eta'}{4\left(1-\frac{d(u_i)}{d(v)}+\alpha'\right)} \right).
\end{align*}
Note that
\begin{equation*}
    \alpha' = \frac{|N(v)\triangle N(u_i)|}{d(v)} \ge \frac{|N(v)\setminus N(u_i)|}{d(v)} \ge 1-\frac{d(u_i)}{d(v)}.
\end{equation*}
So
\begin{align*}
    \Pr[k' \le t] & \le \exp\left( -\frac{(\alpha'-\alpha(1-\gamma))^2\eta'}{8\alpha'} \right) \\
    & \le \exp\left( -\frac{(\alpha-\alpha(1-\gamma))^2\eta'}{8\alpha} \right)  \tag{$\alpha'\ge\alpha$} \\
    & = \exp\left( -\frac{\alpha\gamma^2\eta'}{8} \right) \\
    & = \frac{\beta\gamma}{2}.  \tag{$\eta'=\frac{8}{\alpha\gamma^2}\ln\frac{2}{\beta\gamma}$}
\end{align*}
In summary, when $|N(v)\triangle N(u_i)| \ge \alpha d(v)$, we have
\begin{equation}
\label{eqn:bad-ui}
    \Pr[k' \le t] \le \frac{\beta\gamma}{2}.
\end{equation}

Assume that $\clean[\alpha,\beta](N(v))=\emptyset$.
Then for each sample $u_i$, with probability at least $\beta$, we have $|N(v)\triangle N(u_i)|\ge\alpha d(v)$. Then, by \cref{eqn:bad-ui}, $k$ is the sum of $\eta$ independent Bernoulli random variables, each being $1$ with probability at least
\begin{align*}
    \beta\left(1-\frac{\beta\gamma}{2}\right) \ge \beta\left(1-\frac{\gamma}{2}\right).
\end{align*}
So, by the Chernoff bound,
\begin{align*}
    \Pr[k\le\beta(1-\gamma)\eta] & \le \exp\left( -\frac{(\beta\gamma\eta/2)^2}{2\beta(1-\gamma/2)\eta} \right) \\
    & = \exp\left( -\frac{\beta\gamma^2\eta}{8-4\gamma} \right) \\
    & \le \frac{1}{d^2(v)},  \tag{$\eta=\frac{16}{\beta\gamma^2}\ln d(v)$}
\end{align*}
which means $\check[\alpha,\beta,\gamma](v)=\text{False}$ with probability at least $1-1/d^2(v)$.

\subsection{Proof of \cref{lem:check-true}}
\label{app:check-true}

For a sample $u_i\in N(v)$, $k'$ is the sum of $\eta'$ independent Bernoulli random variables, each with probability $\frac{|N(v)\setminus N(u_i)|}{d(v)}$ being $1$. So,
\begin{align*}
    \E[k'] & = \frac{|N(v)\setminus N(u_i)|}{d(v)}\eta' \\
    & = \frac12\frac{d(v)-d(u_i)+|N(v)\triangle N(u_i)|}{d(v)}\eta'.
\end{align*}
Assume that $|N(v)\triangle N(u_i)|=\alpha'd(v)$ for some $\alpha'\le\alpha(1-2\gamma)$, then
\begin{equation*}
    \E[k'] = \frac12\left(1-\frac{d(u_i)}{d(v)}+\alpha'\right)\eta'.
\end{equation*}
Let $$t=\frac12\left(1-\frac{d(u_i)}{d(v)}+\alpha(1-\gamma)\right)\eta' > \E[k'].$$
By the Chernoff bound,
\begin{align*}
    \Pr[k' \ge t] & \le \exp\left(-\frac{(t-\E[k'])^2}{t+\E[k']}\right) \\
    & = \exp\left( -\frac{(\alpha(1-\gamma)-\alpha')^2\eta'}{2\left(2\left(1-\frac{d(u_i)}{d(v)}\right)+\alpha(1-\gamma)+\alpha'\right)} \right).
\end{align*}
Note that
\begin{equation*}
    \alpha' = \frac{|N(v)\triangle N(u_i)|}{d(v)} \ge \frac{|N(v)\setminus N(u_i)|}{d(v)} \ge 1-\frac{d(u_i)}{d(v)}.
\end{equation*}
So
\begin{align*}
    \Pr[k' \ge t] & \le \exp\left( -\frac{(\alpha(1-\gamma)-\alpha')^2\eta'}{2(\alpha(1-\gamma)+3\alpha')} \right) \\
    & \le \exp\left( -\frac{(\alpha(1-\gamma)-\alpha(1-2\gamma))^2\eta'}{2(\alpha(1-\gamma)+3\alpha(1-2\gamma))} \right)  \tag{$\alpha'\le\alpha$} \\
    & = \exp\left( -\frac{\alpha\gamma^2\eta'}{8-14\gamma} \right) \\
    & \le \frac{\beta\gamma}{2}.  \tag{$\eta'=\frac{8}{\alpha\gamma^2}\ln\frac{2}{\beta\gamma}$}
\end{align*}
In summary, when $|N(v)\triangle N(u_i)| \le \alpha(1-2\gamma) d(v)$, we have
\begin{equation}
\label{eqn:good-ui}
    \Pr[k' \ge t] \le \frac{\beta\gamma}{2}.
\end{equation}

Assume that $\clean[\alpha(1-2\gamma),\beta(1-2\gamma)](N(v))\neq\emptyset$.
Then for each sample $u_i$, with probability at least $1-\beta(1-2\gamma)$, we have $|N(v)\triangle N(u_i)|\le\alpha(1-2\gamma) d(v)$. Then, by \cref{eqn:good-ui}, $k$ is the sum of $\eta$ independent Bernoulli random variables, each being $1$ with probability at most
\begin{align*}
    (1-\beta(1-2\gamma))\frac{\beta\gamma}{2}+\beta(1-2\gamma) \le \beta\left(1-\frac{3\gamma}{2}\right).
\end{align*}
So, by the Chernoff bound,
\begin{align*}
    \Pr[k\ge\beta(1-\gamma)\eta] & \le \exp\left( -\frac{(\beta\gamma\eta/2)^2}{\beta(2-5\gamma/2)\eta} \right) \\
    & = \exp\left( -\frac{\beta\gamma^2\eta}{8-10\gamma} \right) \\
    & \le \frac{1}{d_v^2},  \tag{$\eta=\frac{16}{\beta\gamma^2}\ln d_v$}
\end{align*}
which means $\check[\alpha,\beta,\gamma](v)=\text{True}$ with probability at least $1-1/d^2(v)$.

\subsection{Proof of \cref{lem:very-good}}
\label{app:very-good}

Imagine that for each vertex $v$, we maintain a value $f(v)$ that is initially $1$.
Whenever $\check[\alpha,\beta,\gamma](v)$ returns False (Line 10) while $v$ is $(\alpha(1-2\gamma),\beta(1-2\gamma))$-good, we multiply $f(v)$ by $\frac{1}{d^2(v)}$, which is (the upper bound of) the error probability in \cref{lem:check-false}.
Note that, as long as $\check[\alpha,\beta,\gamma](v)$ returns True while $v$ is $(\alpha(1-2\gamma),\beta(1-2\gamma))$-good, $v$ will be removed before the next $\check$. So, for any $x\in(0,1)$, the probability that $f(v)$ reaches at most $x$ is at most $x$.

Consider the case that \cref{alg:atom-finding} reports that there are no good clusters (Line 15) but there is an $\eps^2$-good-on-average cluster $C$.
By \cref{lem:hit-good-cluster}, some vertex $v\in C$ has been hit $\delta d(v)$ times when its degree was less than $(1+\delta)d(v)$. These hits will add at least $\frac{2\ln n}{\ln d(v)}$ copies of $v$ to $Q$ in total. Since we reach Line 15, $Q$ is empty, which means all these copies have been removed from $Q$. By \cref{lem:hit-good-cluster} again, when each copy was removed, $v$ is $(\alpha(1-2\gamma),\beta(1-2\gamma))$-good, which means $\check[\alpha,\beta,\gamma](v)$ returned False (otherwise $v$ has been removed). So
$$f(v) \le \left(\frac{1}{d^2(v)}\right)^{\frac{2\ln n}{\ln d(v)}} = \frac{1}{n^4}.$$

In other words, if no $f(v)$ reaches at most $\frac{1}{n^4}$, then \cref{alg:atom-finding} never misses $\eps^2$-good-on-average clusters. By a union bound over all vertices, this happens with probability at least $1-\frac{1}{n^3}$.

\subsection{Proof of \cref{lem:time-atom-finding}}
\label{app:time-atom-finding}

The total running time is dominated by the calls of $\check$ and $\clean$.

\paragraph{Running time of $\check$.} For each vertex $v$, let $d_0(v)$ denote its initial degree. During the whole process, for each $i\in[1,(1+\delta)d_0(v)]$, we add at most $$\left\lceil\frac{2(1+\delta)\ln n}{\delta i \ln\frac{i}{1+\delta}}\right\rceil = O\left(\frac{\log n}{i\log i}+1\right)$$ copies of $v$ to $Q$ when $d(v)=i$.\footnote{Strictly speaking, initial copies are all added when $d(v)=d_0(v)$, but we can imagine that they are added when $d(v)$ decreases from $(1+\delta)d_0(v)$ to $d_0(v)$.}
Since $d(v)$ cannot increase, these copies can only be removed from $Q$ when $d(v) \le i$. So, calling $\check$ for each of them takes $O(\log i)$ time. Thus, the total running time for checking $v$ is
\begin{equation*}
    \sum_{i=1}^{(1+\delta)d_0(v)} O\left(\frac{\log n}{i}+\log i\right) = O(\log n\log d_0(v) + d_0(v)\log d_0(v)) = O(d_0(v)\log n).
\end{equation*}
Therefore, the total running time for all $\check$ is $O(m\log n)$.

\paragraph{Running time of $\clean$.} Each $\clean$ takes $O(d^2(v))$ time. If it returns a non-empty cluster, then this cluster contains $\Omega(d(v))$ vertices, each of degree $\Omega(d(v))$. These vertices and their incident edges will then be deleted, so we can charge the running time to these edges. Thus, all non-empty $\clean$ take $O(m)$ time in total.

If $\clean$ returns $\emptyset$, then the prior probability that $v$ passes $\check$ is at most $1/d^2(v)$ by \cref{lem:check-false}. So, the expected running time for empty $\clean$ is $O(1)$ per copy in $Q$, which is then dominated by the running time for $\check$.

\section{Missing Proofs in \cref{sec:pivot-with-good-cluster}}

\subsection{Proof of \cref{lem:atom-expansion}}
\label{app:atom-expansion}

In the beginning, the number of edges between $K$ and $V\setminus K$ is at most $\eps'|K|^2$ since $K$ is $\eps'$-good.
As long as $\beta_v > \alpha_v$, inserting $v$ into $K$ can only reduce this number and increase the size of $K$. So it remains true after expansion.

Now we prove that every vertex inserted into $K$ is in $C'$. Suppose for contradiction that this is not true. Let $v$ be the first vertex inserted into $K$ that is not in $C'$. Let $K_v$ denote the set $K$ at the moment just before inserting $v$. Then $K_v \subseteq C'$ and $|C'\setminus K_v| \le 2\eps'|K_v|$ by \cref{lem:C'-minus-K} (and the fact that we only expand $K$). If we move $v$ into $C'$ in the optimal clustering, there are
$$|N(v)\cap C'| \le |N(v)\cap K_v| = \beta_v|K_v|$$
violated pairs becoming non-violated. Meanwhile, the number of new violated pairs is at most
\begin{align*}
    |N(v) \triangle C'| & \le |N(v)\triangle K_v| + |K_v\triangle C'| \\
    & \le |K_v\setminus N(v)| + |N(v)\setminus K_v| + 2\eps'|K_v| \\
    & = (1-\beta_v+\alpha_v+2\eps')|K_v|.
\end{align*}
Since $\beta_v>1-\beta_v+\alpha_v+2\eps'$, we will improve the optimal clustering, which is a contradiction.

\subsection{Proof of \cref{lem:atom-our-cost}}
\label{app:atom-our-cost}

For a set of pairs of vertices $S$, let $\cost(S)$ denote the cost we pay on these pairs in this step. For a set of vertices $U$, let $\cost(U)$ denote $\cost\left(\binom{U}{2}\right)$.
Our total cost in this step can be written as:
$$\cost(V) = \cost(K) + \cost(C\setminus K) + \cost(V\setminus C) + \cost(K\times(C\setminus K)) + \cost(K\times(V\setminus C)) + \cost((C\setminus K)\times(V\setminus C)),$$
where
$$\cost(K) = |K^-|,$$
$$\cost(V\setminus C)=0,$$
\begin{align*}
    \E[\cost(K\times(V\setminus C))] & = \sum_{v\in(V\setminus K)} \Pr[v\in (V\setminus C)] \cdot \E[\cost(K\times\{v\}) \mid v\in(V\setminus C)] \\
    & = \sum_{v\in(V\setminus K)} (1-p_v)|N(v)\cap K| \\
    & = \sum_{v\in(V\setminus K)} (1-p_v)\beta_v|K|,
\end{align*}
\begin{align*}
    \E[\cost(K\times(C\setminus K))] & = \sum_{v\in(V\setminus K)} \Pr[v\in (C\setminus K)] \cdot \E[\cost(K\times\{v\}) \mid v\in(C\setminus K)] \\
    & = \sum_{v\in(V\setminus K)} p_v|K\setminus N(v)| \\
    & = \sum_{v\in(V\setminus K)} p_v(1-\beta_v)|K|,
\end{align*}
\begin{align*}
    \E[\cost((C\setminus K)\times (V\setminus C))] & = \sum_{v\in(V\setminus K)} \Pr[v\in (C\setminus K)] \cdot \E[\cost(\{v\}\times (V\setminus C)) \mid v\in(C\setminus K)] \\
    & = \sum_{v\in(V\setminus K)} p_v \cdot \E[|N(v)\setminus C| \mid v\in(C\setminus K)] \\
    & \le \sum_{v\in(V\setminus K)} p_v|N(v)\setminus K|  \tag{$K\subseteq C$} \\
    & = \sum_{v\in(V\setminus K)} p_v\alpha_v|K|,
\end{align*}
\begin{align*}
    \E[\cost(C\setminus K)] & \le \E\left[\binom{|C\setminus K|}{2}\right] \\
    & = \frac12\sum_{v\in(V\setminus K)}\Pr[v\in(C\setminus K)] \cdot \E[|C\setminus(K\cup\{v\})|] \\
    & \le \frac12\sum_{v\in(V\setminus K)}p_v \cdot \E[|C\setminus K|] \\
    & \le \frac12\sum_{v\in(V\setminus K)} p_v\eps'|K|.  \tag{\cref{lem:C-minus-K}}
\end{align*}
So
$$\E[\cost(V)] \le |K^-| + \sum_{v\in(V\setminus K)} \left( p_v\left(1-\beta_v+\alpha_v+\frac{\eps'}{2}\right)+(1-p_v)\beta_v \right)|K|.$$

\subsection{Proof of \cref{lem:atom-opt-cost}}
\label{app:atom-opt-cost}

Let $\opt$ be the optimal clustering before we take $C$. Let $\opt-C$ denote the clustering obtained by removing the vertices in $C$ from the corresponding clusters in $\opt$.
After taking $C$, $\opt-C$ is a valid clustering in the remaining graph, whose cost is at least the cost of the new optimal clustering. So, it suffices to compare our cost with the difference between the costs of $\opt$ and $\opt-C$.

For a set of pairs of vertices $S$, let $\loss(S)$ denote the cost of $\opt$ on the pairs in $S$ that vanish in this step. For a set of vertices $U$, let $\loss(U)$ denote $\loss\left(\binom{U}{2}\right)$.
Recall that $C'$ is the cluster in $\opt$ that contains $K$.
The decrease of the cost of the optimal clustering is at least:
$$\loss(V) = \loss(K) +\loss(C'\setminus K) + \loss(V\setminus C') + \loss(K\times(C'\setminus K)) + \loss(K\times(V\setminus C')) + \loss((C'\setminus K)\times(V\setminus C')),$$
where
$$\loss(K) = |K^-|,$$
$$\loss(C'\setminus K) \ge 0,$$
$$\loss(V\setminus C') \ge 0,$$
\begin{align*}
    \loss(K\times(V\setminus C')) & = \sum_{v\in(V\setminus C')} \loss(K\times\{v\}) \\
    & = \sum_{v\in(V\setminus C')} |N(v)\cap K| \\
    & = \sum_{v\in(V\setminus C')} \beta_v|K|,
\end{align*}
\begin{align*}
    \loss(K\times(C'\setminus K)) & = \sum_{v\in(C'\setminus K)} \loss(K\times\{v\}) \\
    & = \sum_{v\in(C'\setminus K)} |K\setminus N(v)| \\
    & = \sum_{v\in(C'\setminus K)} (1-\beta_v)|K|,
\end{align*}
\begin{align*}
    \E[\loss((C'\setminus K)\times (V\setminus C'))] & = \sum_{v\in(C'\setminus K)} \Pr[v\in C]\cdot\E[\loss(\{v\}\times(V\setminus C')) \mid v\in C] \\
    & = \sum_{v\in(C'\setminus K)} p_v |N(v)\setminus C'| \\
    & \ge \sum_{v\in(C'\setminus K)} p_v (|N(v)\setminus K|-|C'\setminus K|) \\
    & \ge \sum_{v\in(C'\setminus K)} p_v (\alpha_v-2\eps')|K|.  \tag{\cref{lem:C'-minus-K}}
\end{align*}
So
$$\E[\loss(V)] \ge |K^-| + \sum_{v\in(V\setminus C')}\beta_v|K| + \sum_{v\in(C'\setminus K)}(1-\beta_v+p_v(\alpha_v-2\eps'))|K|.$$

\subsection{Proof of \cref{eqn:atom-ratio-1a,eqn:atom-ratio-2a}}
\label{app:pivot-with-good-cluster}

Recall that we want to show
\begin{equation}
\label{eqn:atom-ratio-1}
    p_v\left(1-\beta_v+\alpha_v+\frac{\eps'}{2}\right)+(1-p_v)\beta_v \le (2+O(\eps'))\beta_v
\end{equation}
and
\begin{equation}
\label{eqn:atom-ratio-2}
    p_v\left(1-\beta_v+\alpha_v+\frac{\eps'}{2}\right)+(1-p_v)\beta_v \le (2+O(\eps'))(1-\beta_v+p_v(\alpha_v-2\eps'))
\end{equation}
for
\begin{equation}
\label{eqn:pv}
    p_v=\frac{\beta_v}{1+\alpha_v}.
\end{equation}
Also, recall that $2\beta_v \le 1+\alpha_v+2\eps'$, which means
\begin{equation}
\label{eqn:pv-1}
    p_v\le\frac12+\eps'.
\end{equation}

By \cref{eqn:pv}, we first have
\begin{align*}
    p_v\left(1-\beta_v+\alpha_v+\frac{\eps'}{2}\right)+(1-p_v)\beta_v = 2(1-p_v)\beta_v+\frac12p_v\eps'.
\end{align*}
Then, to prove \cref{eqn:atom-ratio-1}, we have:
\begin{align*}
    \left(2+\frac{\eps'}{2}\right)\beta_v - 2(1-p_v)\beta_v - \frac12p_v\eps'
    & = 2p_v\beta_v + \frac12\eps'(\beta_v-p_v) \\
    & \ge 0.  \tag{\cref{eqn:pv}}
\end{align*}
To prove \cref{eqn:atom-ratio-2}, we have:
\begin{align*}
    & \left(2+\frac{7\eps'(1+2\eps')}{2(1-2\eps')^2}\right)(1-\beta_v+p_v(\alpha_v-2\eps')) - 2(1-p_v)\beta_v - \frac12p_v\eps' \\
    = \ & p_v\left(2\alpha_v+2\beta_v-\frac72\eps'\right) + (2-4\beta_v) + \frac{7\eps'(1+2\eps')}{2(1-2\eps')^2}(1-\beta_v+p_v(\alpha_v-2\eps')) \\
    = \ & p_v\left(-2+2\beta_v-\frac72\eps'\right) + (2-2\beta_v) + \frac{7\eps'(1+2\eps')}{2(1-2\eps')^2}(1-p_v(1+2\eps'))  \tag{\cref{eqn:pv}} \\
    = \ & \frac72\eps'\left(\frac{1+2\eps'}{(1-2\eps')^2}(1-p_v(1+2\eps'))-p_v\right) + (2-2\beta_v)(1-p_v) \\
    \ge \ & \frac72\eps'\left(\frac{1+2\eps'}{(1-2\eps')^2}(1-p_v(1+2\eps'))-p_v\right)  \tag{$\beta_v,p_v \le 1$} \\
    \ge \ & \frac72\eps'\left(\frac{1+2\eps'}{(1-2\eps')^2}\left(1-\frac12(1+2\eps')^2\right)-\frac12(1+2\eps')\right)  \tag{\cref{eqn:pv-1}} \\
    = \ & 0.
\end{align*}

\end{document}